\pdfoutput=1

\documentclass[11pt]{amsart}
\usepackage{amssymb,latexsym,amsmath}
\usepackage[pdftex]{graphicx}
\usepackage{graphics}
\usepackage{harvard,enumerate,amssymb}
\usepackage{blkarray}
\usepackage{bm}
\usepackage[utf8]{inputenc}
\usepackage{arydshln}
\usepackage{xcolor}

\newtheorem{thm}{Theorem}
\newtheorem{cor}{Corollary}

\theoremstyle{definition}

\numberwithin{equation}{section}

\newcommand{\argmax}{\operatornamewithlimits{argmax}}

\newcommand{\cB}{\mathcal{B}}

\begin{document}
\title[]{Nonparametric Counterfactuals in Random Utility Models}
\author[Kitamura]{Yuichi Kitamura$^*$}
\address{Cowles Foundation for Research in Economics, Yale University, New
Haven, CT 06520.}
\email{yuichi.kitamura@yale.edu}
\author[Stoye]{J\"org Stoye$^{**}$}
\address{Department of Economics, Cornell University, Ithaca, NY 14853.}
\email{stoye@cornell.edu}
\date{This Version: February 20, 2019}
\thanks{Stoye acknowledges support from the National Science Foundation under grants SES-1260980 and SES-1824375.}

\begin{abstract}
We bound features of counterfactual choices in the nonparametric random utility model of demand \cite{KS}, i.e. if observable choices are repeated cross-sections and one allows for unrestricted, unobserved heterogeneity. In this setting, tight bounds are developed on counterfactual discrete choice probabilities and on the expectation and c.d.f. of (functionals of) counterfactual stochastic demand.
\end{abstract}

\maketitle

\section{Introduction}

Consider the random utility model of demand analyzed in \citeasnoun{mcfadden-richter}, \citeasnoun{McFadden05}, and \citeasnoun{KS}: Repeated cross-sections of demand are observed on a finite sequence of budgets; the maintained assumption is that these cross-sections are of a population of individually rational (in the sense of maximizing utility) individuals; however, one does not substantively restrict utility functions nor their distribution, that is, one allows for unrestricted and possibly infinite dimensional unobserved heterogeneity.

\citeasnoun{McFadden05} characterizes the empirical content of this model in population. We build on his results to provide tight bounds on the distribution of counterfactual demand, i.e. of stochastic demand on as yet unobserved budgets, as well as explicit bounds on the expectation and c.d.f. of linear functions of demand vectors, e.g. demand for a specific good. Many of these bounds turn out to be the values of linear programs, hence are easy to compute even in moderately high dimensional applications.\footnote{Some of these results were reported in section 9.2 of \citeasnoun{KS13} and implemented at the time. We make them available not least because other work already built on them \cite{Adams16,Manski14}. Code is available from the authors. See also \citeasnoun{Adams16}, \citeasnoun{Hubner}, and \citeasnoun{smeulders} for recent results on computational implementation.} We next describe the setup and recall an important characterization of stochastic rationalizability, then provide the bounds, and close by mentioning some extensions.

\section{Stochastic Rationalizability}\label{sec:rational}

We use notation from \citeasnoun{KS}. There are $J$ observed budgets $\{\mathcal{B}_j\}_{j=1}^J, J \in \bf N$, each characterized by price vectors $p_j \in \mathbf{R}^K_+$, where expenditure is normalized to $1$:
\begin{equation*}
\mathcal{B}_{j}\equiv \{y \in \mathbf{R}^K_+:p_j'y=1\},\quad j=1,...,J.
\end{equation*}%
Suppose that we know a \textit{stochastic demand system}
\begin{equation*}
P_j(x)\equiv \Pr(y(p_j)\in x),\quad x\subset \mathbf{R}_{+}^{K}
\end{equation*}%
for $j=1,\dots,J$, where the random variable $y(p_j)$ is demand on budget $\mathcal{B}_j$. This collection of distributions is rationalizable by a random utility model if there exists a distribution $P_u$ over locally nonsatiated (for simplicity) utility functions $u:\mathcal{\mathbf{R}}_+^K \mapsto \mathbf{R}$ s.t. 
\begin{equation*}
P_j (x)=\int \! 1\Bigl\{\argmax_{y\in \mathbf{R}_+^K:p_j'y=1}u(y) \in x\Bigr\} \mathrm{d}P_u, \quad x\subseteq \mathcal{B}_j,j=1,...,J. 
\end{equation*}

Our motivation is demand estimation from repeated cross-section with unobserved heterogeneity, but the model has also been used to describe choices made by an individual with random utility. We next recall a succinct description of its empirical content.

Let $\mathcal{X} \equiv \{x_1,...,x_I\}$ be the coarsest partition of $\cup _{j=1}^J\mathcal{B}_j$ such that for any $i \in \{1,...,I\}$ and $j \in \{1,...,J\}$, $x_i$ is either completely on, completely strictly above, or completely strictly below budget plane $\mathcal{B}_j$. Equivalently, any $y_1,y_2 \in \cup _{j=1}^J\mathcal{B}_j$ are in the same element of the partition iff $\text{sg}(p_j'y_1-1)=\text{sg}(p_j'y_2-1)$ for all $j=1,...,J$. Elements of $\mathcal{X}$ will be called \textit{patches}. Each budget can be uniquely expressed as union of patches; the number of patches that jointly comprise budget $\mathcal{B}_j$ will be called $I_j$. For future reference, we emphasize that any patch is the intersection of finitely many open or closed half spaces and therefore its closure (though not necessarily the patch itself) is a finite polytope.

An important insight of the aforecited papers is that stochastic rationalizability constrains the aggregate choice probabilities \textit{of} patches, but not at all the distribution of demand \textit{on} any patch. Intuitively, this is because all choices that are on the same patch generate the same revealed preference information. Formally, let the \textit{vector representation} of $(\mathcal{B}_1,\dots,\mathcal{B}_J)$ be the $\bigl(\sum_{j=1}^J I_j\bigr)$-vector $(x_{1|1},\dots,x_{I_1|1},x_{1|2},\dots,x_{I_J|J})$, where $(x_{1|j},\dots,x_{I_j|j})$ lists all patches comprising $\mathcal{B}_j$ in arbitrary but henceforth fixed order.\footnote{Note that elements of $\mathcal{X}$ that appear as components of distinct budgets make corresponding repeat appearances, under different labels, in the vector representation.} Let the \textit{vector representation} of $(P_1,\dots,P_J)$ be the $\bigl(\sum_{j=1}^{J}I_{j}\bigr)$-vector $\pi \equiv (\pi_{1|1},\dots,\pi_{I_1|1},\pi_{1|2},\dots,\pi_{I_J|J})$, where $\pi_{i|j} \equiv P_j(x_{i|j})$. Next, note that any rationalizable nonstochastic demand system can be thought of as a degenerate stochastic demand system with binary vector representation. A stochastic demand system is rationalizable iff it is a mixture of such rationalizable nonstochastic demand systems because the latter can be thought of as representing choice types in the population. But due to the discretization of the choice universe into patches, there are only finitely many such types. Collect their vector representations in the $H < \infty$ columns of the \textit{rational demand matrix} $A$: see \citeasnoun{KS} (in particular Definition 3.5 and discussions in Sections 3.2 - 3.4). Then we have:\footnote{The statement follows \citeasnoun{KS}, who also prove it, provide algorithms for computing $A$, and point out that $\nu \in \Delta^{H-1}$ can be conveniently weakened to $\nu \geq 0$. However, the discretization step is clearly anticipated in \citeasnoun{McFadden05}, and the result was otherwise proved in \citeasnoun{mcfadden-richter}. See also \citeasnoun{Stoye18}.}

\begin{thm}\label{T1}
The stochastic demand system $(P_1,\dots,P_J)$ is rationalizable if, and only if, its vector representation $\pi$ fulfills $\pi =A\nu$ for some $\nu \in \Delta^{H-1}$. Here, $\Delta^{H-1}$ is the unit simplex in ${\bf R} ^H$.
\end{thm}

\section{Bounds on Counterfactuals}

We next take a rationalizable stochastic demand system $(P_1,\dots,P_J)$ as given and ask what discipline it places on
$$y(p_0)  := \argmax_{y\in \mathbf{R}_+^K:p_0'y=1}  u(y), \quad u \sim P_u,
$$
the stochastic demand at some counterfactual budget $\cB_0$ corresponding to counterfactual price $p_0$.\footnote{For the very special case of $K=2$, \citeasnoun{HS15} provide closed-form bounds. \citeasnoun{BBC08} and many others provide bounds under slightly stronger, e.g. aggregation, assumptions.} As with nonstochastic demand, this discipline will typically take the form of bounds, although these are now on a distribution. They are tightly related to testing rationalizability because a distribution $P_0$ of demands on $\cB_0$ is inside the bounds iff $(P_0,...,P_J)$ are jointly rationalizable; thus, Theorem \ref{T1} implies an exact characterization of bounds on $P_0$ implied by knowledge of $(P_1,...,P_J)$.
We will now formally state this characterization.
 
Recall that the matrix $A$ in Theorem \ref{T1} is obtained for the set of observed budgets $(\mathcal B_1,...,\mathcal B_J)$.  We can apply the same algorithm to the augmented set of budgets  $(\mathcal B_0, \mathcal B_1,...,\mathcal B_J)$ to obtain patches on it and its vector representations: for completeness we write them 
\begin{equation}\label{eq:vec}
(x^*_{1|0},...,x^*_{I^*_0|0},x^*_{1|1},...,x^*_{I^*_J|J}).
\end{equation}
The patches for the original system $(\mathcal B_1,...,\mathcal B_J)$ remain unchanged in the augmented system if they do not intersect with $\mathcal B_0$.  Therefore if $\mathcal B_{j'} \cap \mathcal B_0=\emptyset$ holds for some $j' \in \{1,\dots,J\}$, then  $I_{j'}^* = I_j$ and    
$$
(x^*_{1|j'},\dots,x^*_{I^*_{j'}|j'}) = (x_{1|j'},\dots,x_{I_{j'}|j'})
$$ 
for such $j'$.  Moreover we can apply the algorithm discussed in Section \ref{sec:rational} to the augmented system  $(\mathcal B_0,\dots,\mathcal B_J)$ to obtain its rational demand matrix $A^* \in {\bf R}^{\left(\sum_{j=1}^{J+1}I_j^*\right) \times {H^*}}$, where $H^* \geq H$. Note that each row of $A^*$ corresponds to a patch in the new vector representation  \eqref{eq:vec}.  Once $A^*$ is obtained, we can define a probability vector $\nu^* \in \Delta^{H^* -1}$, now defined over the columns of $A^*$, and the choice probability vector $\pi^*$ for the patches  $\{ x^*_{1|0},...,x^*_{I^*_0|0},x^*_{1|1},...,x^*_{I^*_J|J}  \}$.
Note that the elements of $\pi^*$ corresponding to $\cup_{j=1}^J \mathcal B_j$ are observed, while the rest remain unobserved: the latter are counterfactual conditional probabilities.  To make this point clear, we write    
\begin{equation*}
A^*=\left[ 
\begin{array}{c}
A_0^* \\ \hdashline[2pt/2pt]
A_1^*%
\end{array}
\right] ,~~\pi^* =\left[ 
\begin{array}{c}
\pi_0^* \\ \hdashline[2pt/2pt]
\pi_1^*
\end{array}
\right],
\end{equation*}
where $A_0^*$ collects rows of $A^*$ that correspond to patches that do not belong to  $\cup_{j=1}^J \mathcal B_j$, $A_1^*$ collects all other patches, and similarly for $\pi^*$. It continues to be the case that $\pi^*$ is rationalizable iff $A^*\nu^*=\pi^*$ for some $\nu^* \in \Delta^{H^*-1}$. However, rather than taking $\pi^*$ to be observed and testing rationalizability, we take $\pi_1^*$ to be observed and $\pi_0^*$ to a vector of counterfactual probabilities to be accordingly constrained by the observed $\pi_1^*$. Formally:\footnote{For a setting like ours except that the universal choice set is finite and ``budgets" are subsets of it, \citeasnoun{Manski07} anticipates Theorem \ref{thm:2}. Our contribution lies in the connection to nonparametric demand, in laying the groundwork for Theorem \ref{thm:3} and its corollaries, and in the accompanying computational as well as statistical machinery.}  
\begin{thm}\label{thm:2}
A distribution $P_0$ is consistent with observed demands $(P_1,...,P_J)$ if, and only if, its implied value of $\pi_0^*$ fulfils
\begin{equation*}
A^* \nu^*=\left[ 
\begin{array}{c}
\pi_0^* \\ \hdashline[2pt/2pt]
\pi_1^*
\end{array}
\right]
\end{equation*}
for some $\nu^* \in \Delta^{H^*-1}$. Here, $\pi_1^*$ takes the value implied by $(P_1,\dots,P_J)$. In particular, the conditional distributions $P_0(\cdot|y \in x_{i|0}^*)$ (for all for all $i=1,\dots,I_0^*$ where this is defined) are not restricted.
\end{thm}
We next explain how this result translates into extremely tractable, best possible bounds on many parameters of interest. Specifically, we have:
\begin{thm}\label{thm:3}

For any known function $g:{\bf R}^K \mapsto {\bf R}$ that is bounded on $\cB_0$, define
\begin{eqnarray}
\underline{g}_{i|0} &\equiv &\inf_{y\in x_{i|0}^*}g(y),\quad 1\leq i\leq I_0^* \label{eq:def_glo}\\
\overline{g}_{i|0} &\equiv &\sup_{y\in x_{i|0}^*}g(y),\quad 1\leq i\leq I_0^*. \label{eq:def_ghi}
\end{eqnarray}%
Then the bounds
\begin{multline}
\min \bigl\{(\underline{g}_{1|0},\dots,\underline{g}_{I_0|0}) A_0^*\nu^*: A_1^*\nu^* =\pi_1^*,\nu^* \in \Delta^{H^*-1} \bigr\}  \\
\leq \mathbb{E}g(y(p_0))  \leq  \label{eq:C2}\\
\max \bigl\{(\overline{g}_{1|0},\dots,\overline{g}_{I_0|0}) A_0^*\nu^*: A_1^*\nu^* =\pi_1^*,\nu^* \in \Delta^{H^*-1} \bigr\} 
\end{multline}
are sharp, i.e. they cannot be improved upon without further information.
\end{thm}
\begin{proof}
By the Law of Iterated Expectations,
\begin{eqnarray}
\mathbb{E}g(y(p_0)) = \sum_{i=1}^{I_0} \pi_{i|0}^* \mathbb{E}(g(y(p_0))|y \in x_{i|0}^*) = (g_{1|0},\dots,g_{I_0^*|0}) \pi_0^*,
\label{eq:expected}
\end{eqnarray}
where $g_{i|0}\equiv \mathbb{E}(g(y(p_0))|y \in x_{i|0}^*)$ if $\pi_{i|0}^* \neq 0$ and otherwise we assign it an arbitrary value.   By inspection of \eqref{eq:def_glo}, \eqref{eq:def_ghi}, and \eqref{eq:expected}, the upper and lower bounds are valid and can be approached arbitrarily closely. Furthermore, if distributions $P_0$ and $Q_0$ are consistent with observable demands $(P_1,\dots,P_J)$, then so is any mixture between them. Hence, all values strictly between the bounds are attained by appropriate mixtures of distributions that approximate the bounds. 
\end{proof}

The proof reveals not only that the bounds are sharp, but also that all intermediate values of $\mathbb{E}g(y(p_0))$ are necessarily attainable. Whether the bounds themselves are attainable depends on whether patches are open or closed in the relevant directions and can only be decided on a case-by-case basis.   

Computing these bounds requires to solve the linear programs in \eqref{eq:C2} and, as an input, the optimization problems in \eqref{eq:def_glo}-\eqref{eq:def_ghi}. The latter are tractable in relevant cases: If $g$ is continuous, the constraint sets can be taken to be the closures of patches, hence finite polytopes. If $g$ is furthermore linear, then  computing the bounds requires only linear programming, though possibly with many constraints. 

We further elaborate this result by more explicitly bounding the expected value and c.d.f. of $z'y(p_0)$, where $z \in {\bf R} ^K$ is a user-specified vector. For example, $z=(1,0,\dots,0)'$ extracts demand for good $1$ and $z=(p_0^{[1]},p_0^{[2]},0,\dots,0)'$ (i.e., the first two components of $p_0$ followed by zeroes) extracts joint expenditure on the first two goods. Theorem \ref{thm:3} then specializes as follows.
\begin{cor}\label{thm:4}
Let
\begin{eqnarray*}
\underline{m}_{i|0}(z) &\equiv &\inf \{z'y:y\in x_{i|0}^*\},\quad 1\leq i\leq I_0^* \\
\overline{m}_{i|0}(z) &\equiv &\sup \{z'y:y\in x_{i|0}^*\},\quad 1\leq i\leq I_0^*.
\end{eqnarray*}
Then the bounds
\begin{multline*}
\min \{(\underline{m}_{1|0}(z),\dots,\underline{m}_{I_0|0}(z)) A_0^*\nu^*: A_1^*\nu^* =\pi_1^*,\nu^* \in \Delta^{H^*-1} \} \\
\leq \mathbb{E}(z'y(p_0))  \leq \\
\max \{(\overline{m}_{1|0}(z),\dots,\overline{m}_{I_0|0}(z)) A_0^*\nu: A_1^*\nu^* =\pi_1^*,\nu^* \in \Delta^{H^*-1} \}
\end{multline*}
are sharp.
\end{cor}
We note that computation of these bounds only requires linear programming. Next, we bound probabilities of arbitrary events and hence also c.d.f.'s.
\begin{cor}\label{thm:5}
For fixed event $x \subseteq \cB_0$, the bounds
\begin{multline*}
\min \biggl\{ \sum_{\substack{ i\in \{1,...,I_0^*\}: \\ x_{i|0}^*\subseteq x }} e_i' A_0^*\nu :A_1^*\nu^* =\pi_1^*,\nu^* \in \Delta^{H^*-1} \biggr\} \\
\leq \Pr(y(p_0) \in x) \leq \\
\max \biggl\{ \sum_{\substack{ i\in \{1,...,I_0^*\}: \\ x_{i|0}^* \cap x \neq \emptyset}} e_i' A_0^*\nu :A_1^*\nu^*= \pi_1^*,\nu^* \in \Delta^{H^*-1} \biggr\}
\end{multline*}
are sharp. Here, $e_i$ is the $i$'th canonical basis vector in ${\bf R} ^{I_0^*}$.

For fixed vector $z \in {\bf R} ^K$, let
\begin{eqnarray*}
\underline{p}_{i|0}(z,t) &\equiv &\bm{1}\{\overline{m}_{i|0}(z)\leq t\} \\
\overline{p}_{i|0}(z,t) &\equiv &\bm{1}\{x_{i|0}^* \cap \{z'y=t\} \neq \emptyset \},
\end{eqnarray*}
noting that $\underline{m}_{i|0}(z)<t \Rightarrow \overline{p}_{i|0}(z,t)=1$ and $\underline{m}_{i|0}(z)>t \Rightarrow \overline{p}_{i|0}(z,t)=0$. 
Then the following bounds on the c.d.f. of $z'y(p_0)$ are sharp:
\begin{multline*}
\min \bigl\{ (\underline{p}_{1|0}(z,t),\dots,\underline{p}_{I_0^*|0}(z,t)) A_0^*\nu^* :A_1^*\nu^* =\pi_1^*,\nu^* \in \Delta^{H^*-1} \bigr\} \\
\leq \Pr(z'y(p_0) \leq t) \leq \\
\max \bigl\{ (\overline{p}_{1|0}(z,t),\dots,\overline{p}_{I_0^*|0}(z,t)) A_0^*\nu^* :A_1^*\nu^* =\pi_1^*,\nu^* \in \Delta^{H^*-1} \bigr\}.
\end{multline*}
\end{cor}
The bounds on the c.d.f. require essentially only linear programming, with a minimal additional check in the finitely many cases where $\underline{m}_{i|0}(z)=t$.\footnote{This case occurs when $\{z'y=t\}$ is a lower supporting hyperplane of patch $x_{i|0}$ but does not intersect it.} They are pointwise but not uniform in $t$; in particular, their upper and lower envelopes do not necessarily describe feasible counterfactual distributions.\footnote{Indeed, they may not even be c.d.f.'s for lack of right-continuity, though they can always be approximated by c.d.f.'s.} Therefore, while the upper and lower envelopes induce bounds on a multitude of more complicated parameters \cite{Stoye10}, those bounds are not in general tight. 

\section{Concluding Remarks}

We conclude by mentioning some connections and extensions.

First, we considered the case of one counterfactual budget for expositional clarity. Bounds on the joint c.d.f. of demand on two counterfactual budgets, or on some linear combination of expected values, are straightforward extensions of the above results. In general, they can be considerably tighter than the Cartesian product of budget-by-budget bounds and also need not include the budget-by-budget minimum and maximum bound. For the case of a single (possibly fictitious, e.g. representative) nonstochastic utility maximizer, see \citeasnoun{Adams16} for a much more extensive analysis in this spirit.

Next, these results naturally extend to finite discrete choice settings, i.e. if choices from distinct subsets $\mathcal{C}_1,\dots,\mathcal{C}_J$ of a finite choice universe $\mathcal{X}$ (the duplication of notation is intended) were observed and choices from another such subset $\mathcal{C}_0$ are to be predicted. In this case, the finitely many elements of $\mathcal{X}$ directly play the role of patches, the conditional distributions on patches are trivial, and Theorem 2 characterizes those  p.m.f.'s of counterfactual random choice that are consistent with observed choice distributions. That said, the analysis in \citeasnoun{Manski07} anticipates Theorem 2 in this setting. Of course, the result also applies to the further specialization where all observed choice sets are binary, as in the ``linear polytope" literature in mathematical psychology \cite{fishburn92}.

Finally, we developed population-level bounds but ignored estimation and inference. To handle this, note that if one takes $(p_0,\dots,p_J)$ and therefore $A^*$ to be known, then all the above bounds maximize or minimize $\gamma A_0^* \nu^*$ for some known vector $\gamma$; it is only $\pi_1^*$ that must be estimated. This is essentially the estimation and inference problem analyzed in Section 4.2 of \citeasnoun{DKQS16}.

\bibliographystyle{econometrica}
\bibliography{mybib}
\end{document}